\documentclass{article}
\usepackage{ijcai19}

\usepackage{times}
\usepackage{soul}
\usepackage{url}
\usepackage[hidelinks]{hyperref}
\usepackage[utf8]{inputenc}
\usepackage[small]{caption}
\usepackage[export]{adjustbox}
\usepackage{algorithm}
\usepackage[noend]{algpseudocode}
\usepackage{amssymb,amsmath,amsthm}
\usepackage{array,multirow}
\usepackage{authblk}
\usepackage{booktabs}
\usepackage{caption}
\usepackage{color}
\usepackage{diagbox}
\usepackage{enumitem}
\usepackage{float}
\usepackage{graphicx}
\usepackage{hyperref}
\usepackage{multirow}
\usepackage{mathbbol}
\usepackage{soul}
\usepackage{subcaption}
\usepackage{times}
\usepackage{tabu}
\usepackage{tablefootnote}
\usepackage[dvipsnames,table,xcdraw]{xcolor}
\algrenewcommand{\algorithmicrequire}{\textbf{Input:}}
\algrenewcommand{\algorithmicensure}{\textbf{Output:}}

\newcolumntype{P}[1]{>{\centering\arraybackslash}p{#1}}
\newcolumntype{M}[1]{>{\centering\arraybackslash}m{#1}}

\newtheorem{definition}{Definition}[]
\newtheorem{example}{Example}[]
\newtheorem{lemma}{Lemma}[]
\newtheorem{proposition}{Proposition}[]

\newtheorem{theorem}{Theorem}[]

\title{Extending Eigentrust with the Max-Plus Algebra}

\author{
Juan Afanador$^1$\and
Nir Oren$^1$\and
Maria Araujo$^2$\and
Murilo Baptista$^2$
 \\
$^1$Department of Computing Science, University of Aberdeen, Scotland, UK\\
$^2$Department of Physics, University of Aberdeen, Scotland, UK\\
}

\begin{document}

\maketitle

\begin{abstract}
Eigentrust is a simple and widely used algorithm, which quantifies trust based on the repeated application of an update matrix to a vector of initial trust values. In some cases, however, this procedure is rendered uninformative. Here, we characterise such situations and trace their origin to the algebraic conditions guaranteeing the convergence of the Power Method. We overcome the identified limitations by extending Eigentrust's core ideas into the Max-Plus Algebra. The empirical evaluation of our max-plus approach demonstrates improvements over Eigentrust.
\end{abstract}

\section{Introduction}
\label{sec:intro}

Eigentrust is a popular approach to the calculation of trust with local information in multi-agent systems \cite{kamvar2003eigentrust}. The functioning of the algorithm presupposes the free flowing of firsthand experiences across networks of agents upon request, thus improving or augmenting the information set of the querying agent. Eigentrust views trust as a normalised score of reliability computed with the querying agent's own experience, and the information gathered from its neighbours. 

The query may be replicated by any other agent within the initial vicinity, and their neighbours' neighbours. This information is aggregated into a single value, where the credibility parameters specified by the agent posing the query, serve as weights. The resulting measure indicates the level of trust placed in the agent about which the query was made.

More generally, whenever a query is issued, Eigentrust stores local trust scores in a matrix which then applies to an initial vector of trust values. An operation repeated until convergence to a fixed-point, in the expectation that the final vector would coincide with the dominant eigenvector of the matrix. This latter vector is, thereby, considered an accurate global ranking. 

Due to its simplicity, theoretical foundations and resulting empirical behaviour, Eigentrust is widely used in domains such as P2P systems \cite{marti2006taxonomy}, Internet-of-Things architectures \cite{azad2017m2m}, and Ad-hoc Sensor Networks \cite{theodorakopoulos2006trust}. While Eigentrust performs particularly well in networks of homogeneous agents such as P2P systems, it is often lacking in more diverse environments. A situation we attempt to overcome.

In this paper we study the performance of Eigentrust in networks with various degrees of connectivity, and describe the cases where the algorithm accurately predicts global trust scores, others where it is somehow inaccurate but useful (e.g., where it can spot malicious behaviour), and those where it may be misleading. We provide a theoretical characterisation of all these cases to build on Eigentrust's core ideas toward a more generally applicable procedure. 

Our goal is to formulate an algorithm that operates across more diverse environments than Eigentrust does. We argue that Eigentrust performs poorly in, precisely, those cases where convergence to the dominant eigenvector does not occur. By framing the trust-measuring problem within a different algebraic structure --- the Max-Plus Algebra --- we are able to obtain informative trust ratings in these situations where Eigentrust fails.

The remainder of this paper is structured as follows. The next section introduces the Eigentrust algorithm and the algebraic assumptions it is founded upon. Section \ref{sec:MaxTrust} describes the Max-Plus Algebra,  and characterises the corresponding trust measuring problem. In Section \ref{sec:eval} we present the results of our approach. The last two sections discuss our findings and provide suggestions for future work.

\section{Eigentrust}
\label{sec:eigentrust}
\subsection{The Eigentrust Algorithm}

Eigentrust considers interactions between pairs of agents and observes the corresponding outcome \cite{kamvar2003eigentrust}. If we let $s_{ij}$ denote the difference between the number of successful and unsuccessful interactions between agents $i$ and $j$, then $c_{ij}=max(s_{ij},0)/\sum_k max(s_{ik},0)$ can be viewed as a normalised measure of trust between $i$ and $j$. Eigentrust assumes that trust is transitive, i.e.,  $i$'s trust in $k$ can be computed from its level of trust in $j$, and $j$'s trust in $k$:

\begin{equation}
    \label{eq:trust_score}
    t_{ik}=\sum_j c_{ij}c_{jk};
\end{equation}

\noindent this property means that trust in every agent within a connected component, can be computed in a similar manner.

The set of $c_{ij}$ values can be represented as a \emph{trust matrix} $\mathbf{C}$ of local trust scores. To capture the transitive nature of trust, $\mathbf{C}$ must be applied to a vector of initial trust values $\mathbf{r}$, depending on the relative position of the agent about whom the query was made. Repeating this multiplication incorporates direct and indirect trust information about all relevant neighbours, into the production of a vector of stable trust values for every agent in the system. 

As discussed in Section \ref{sec:eigentrust}, the initial vector against which multiplication takes place can (effectively) be random as long as the operation is carried out to the right of $\mathbf{r}$. This process guarantees that the final trust ranking converges to the ``true" distribution of trust values as given by the matrix's dominant eigenvector ---provided it is unique. The procedure by which the dominant eigenvector is calculated is called the \emph{Power Method} \cite{chu1993multivariate}. Algorithm \ref{algo:eigentrust} summarises the Eigentrust approach to quantifying trust. We note that the final trust value associated with agent $i$ occurs at index $i$ within the eigenvector.

\begin{center}
	\begin{minipage}{\columnwidth}
		\centering
		\begin{algorithm}[H]
			\caption{Eigentrust}
			\label{algo:eigentrust}
			\begin{algorithmic}[1]
				\State $\mathbf{t}^{(0)} \gets \mathbf{r}$ \label{alg:eigentrust:op0}
				\Repeat
				\State
				$\mathbf{t}^{(k+1)}\gets  \mathbf{C}^{T}\mathbf{t}^{(k)}$ \label{alg:eigentrust:op1}
				\State $\delta\gets |\mathbf{t}^{(k+1)}-\mathbf{t}^{(k)}|$ \label{alg:eigentrust:op2}
				\Until{$\delta<\epsilon$}
			\end{algorithmic}
		\end{algorithm}
	\end{minipage}
\end{center}

\subsection{The Algebraic Conditions Behind Eigentrust } 

Since Eigentrust will only operate correctly when the Power Method converges, we now consider one case where Eigentrust performs poorly due to the nonexistence of the dominant eigenvector, and the outright unattainability of convergence. We then characterise two situations where the Power Method is fully applicable. 

\subsubsection{Diagonalisable Square Matrices}

Suppose the matrix $\mathbf{C}$ satisfies the following two conditions:

\begin{enumerate}
	\item  $\mathbf{C}$ is diagonalisable, i.e., there exists an invertible matrix $\mathbf{P}$ such that
	\begin{equation}
	\mathbf{D}=\mathbf{P}^{-1}\mathbf{C}\mathbf{P}
	\end{equation}\label{diagonalisable}
	where $\mathbf{D}$ is a diagonal matrix; and 
	\item $\mathbf{C}$ has a dominant eigenvalue  $\lambda_0$. That is, if $\lambda_{0},\dots, \lambda_{n-1}$ are the eigenvalues of $\mathbf{C}$, then it is the case that 
	$
	|\lambda_{0}|>|\lambda_{i}|
	$
	for $i=1,\dots,n-1$. The eigenvector associated with the dominant eigenvalue is termed the \emph{dominant eigenvector}.
\end{enumerate}

In such a situation, the Power Method converges to the dominant eigenvector of $\mathbf{C}$. Note that if the matrix $\mathbf{C}$ satisfies the two conditions above, then these results also apply to its transpose.

\begin{example}
	\label{ex:e2}
    The following matrix depicts a situation where agents $a_0$, $a_1$ and $a_2$ have limited information about one another.
	\[
	A = \bordermatrix{~ & a_0 & a_1 & a_2 \cr
            		   &  0.75 & 0 & 0.25 \cr 
                       & 0 & 1 & 0 \cr
                       & 0.25 & 0 & 0.75 \cr}
	\]
	Despite being diagonalisable, the three eigenvalues of $\mathbf{A}$ are $\lambda_{0}=\lambda_{1}=1$ and $\lambda_{2}=0.5$. As $\mathbf{A}$ does not have a dominant eigenvalue, the convergence of the Power Method cannot be guaranteed. The lack of connectivity between agents induces an unstable outcome. 
	
	While Eigentrust could --- potentially --- be applied to each connected component in such a graph, this would require previous knowledge of the network structure. Furthermore, in dynamic situations, questions arise as to how trust across connected components should be merged when the topology of the network changes.

\end{example}

\begin{example}
	\label{ex:e1}

Consider the following trust matrix.
	\[
	\mathbf{B}=\bordermatrix{~ & a_0 & a_1 & a_2 \cr
                            	& 0.1 & 0.55 & 0.35 \cr
                            	& 0 & 0.8 & 0.2 \cr
                            	& 0 & 0 & 1 \cr}
	\]
	\noindent This matrix  is upper triangular with distinct diagonal entries, and has eigenvalues $\lambda_{0}=1$, $\lambda_{1}=0.8$ and $\lambda_{2}=0.1$. Therefore, for a random vector $\mathbf{v}\in \mathbb{R}^{3}$ it is the case that 
	$\lim_{k\longrightarrow \infty}\left(\mathbf{B}^{T}\right)^{k}\cdot\mathbf{v} = \vec{\pi}$, where $\vec{\pi}=(0,0,1)^{T}$ is the eigenvector associated with the dominant eigenvalue $\lambda_{0}=1$. Therefore, according to this matrix and by extension to Eigentrust, only agent $a_2$ can be trusted. A decision maker thus views all other options as equally irrelevant, which may not be informative enough in some applications; for it is only restating the same information already conveyed by the trust matrix.
\end{example}

\subsubsection{Positive Stochastic Matrices}

	Suppose $\mathbf{C}$ is a square, positive and stochastic matrix.
	Using the Perron-Frobenius Theorem \cite{pillai2005perron} and the Jordan decomposition of $\mathbf{C}$, it is possible to show that there exists a unique dominant eigenvector and the limit $\lim_{k\longrightarrow \infty}\left(\mathbf{C}^{T}\right)^{k}\cdot\mathbf{v}$ exists and converges to the same value for any initial random vector $\mathbf{v}$. Note that the existence of a dominant eigenvector is a consequence of the three previous conditions. Lemma \ref{lem:maximal_eigenvalue} and Proposition \ref{prop: eigentrust pr-eigenvector}, which consider a more general case, are built upon these observations.

\begin{example}
	\label{ex:e3}
	Consider the following matrix.
	\[
	\mathbf{C}=\bordermatrix{~ & a_0 & a_1 & a_2 \cr
                                & 0.15 & 0.55 & 0.3 \cr 
                            	& 0.41 & 0.53 & 0.06 \cr
                            	& 0.18 & 0.62 & 0.2 \cr}
	\]
	Its dominant eigenvalue and eigenvector are $\lambda_0=1$ and $\vec{1}$. Thus, for any random ranking $\mathbf{v}$ we have $\left(\mathbf{C}^{T}\right)^{k}\cdot\mathbf{v} \longrightarrow \vec{\pi}$, as $k\longrightarrow \infty$, where $\vec{\pi} = (0.3, 0.6, 0.1)^{T}$, implying that the interactions with agents $a_0$ and $a_1$ will be more frequent. Here, the positive, square and stochastic nature of such a matrix ensures that Eigentrust works as expected.
\end{example}

Upper triangularity with distinct diagonal entries in Example \ref{ex:e1} guarantees diagonalisability, hence convergence of the Power Method. In Example \ref{ex:e2}, however, convergence is not attained despite diagonalisability via symmetry. Convergence in Example \ref{ex:e3} is guaranteed via the  Perron-Frobenius Theorem for positive (stochastic) matrices. We argue that Eigentrust's performance depends on whether a dominant unique eigenvector does, or does not exist, and that in many situations this is not the case. 

The Perron-Frobenius Theorem has been generalised to cater for non-negative and irreducible matrices \cite{seneta2006non}. Our aim is to build on these results in the context of the Max-Plus Algebra, providing a new trust-measuring procedure also applicable to reducible matrices. In the following we will introduce a version of the Perron-Frobenius Theorem for irreducible matrices, and consider its implications to the Eigentrust algorithm. 

\subsubsection{Non-negative Stochastic Matrices}

\begin{definition}[Irreducible matrix \cite{seneta2006non}]
	An n $\times$ n matrix $\mathbf{A} $ is said to be \emph{irreducible} if there is no permutation of coordinates such that:
	
	\begin{equation}
	\label{reducible}
	\mathbf{PAP}^{T}=\left(
	\begin{array}{*2{c}}
	\mathbf{A}_{11} &   \mathbf{A}_{12}   \\
	\mathbf{0}   &   \mathbf{A}_{22}   \\
	\end{array}
	\right)
	\end{equation}
	where $\mathbf{P}$ is an $n\times n$ permutation matrix with each row and column having a single entry equal to one and the rest full of zeros; while $\mathbf{A}_{11}$ and $\mathbf{A}_{22}$ are non-trivial (i.e., their size is greater than 0) square matrices. In other words, an irreducible matrix cannot be converted into a block upper-triangular matrix via simultaneous row/column permutations. A matrix is \emph{reducible} if it is not irreducible.
\end{definition}

\begin{theorem}[Perron - Frobenius Theorem \cite{pillai2005perron}]
    \label{theo:perron_frobenius}
	If $\mathbf{C}=(c_{ij})$ is an $n\times n$ irreducible non-negative matrix with spectral radius $\rho(\mathbf{C})=\lambda_{0}$ , then:
	
	\begin{enumerate}
		\item $\lambda_0\in\mathbb{R}^{+}$ is a simple eigenvalue of  $\mathbf{C}$, called the Perron-Frobenius eigenvalue.
		\item $\lambda_0$ can be associated with unique (up to a constant) and strictly positive left and right eigenvectors. 
	\end{enumerate}
\end{theorem}

\begin{lemma}[\cite{berman1994nonnegative}]
    \label{lem:maximal_eigenvalue}
	Let $\mathbf{C}$ denote an irreducible non-negative stochastic matrix. The set of eigenvalues of  $\mathbf{C}$ has a maximal eigenvalue equal to $1$, and an associated left eigenvector describing a probability distribution $\vec{\pi}$ over the set of interacting nodes.
\end{lemma}

Given the properties of the trust matrix described before and Lemma \ref{lem:maximal_eigenvalue}, it trivially follows that:

\begin{proposition} 
	\label{prop: eigentrust pr-eigenvector}
	If $\mathbf{C}$ is an irreducible non-negative trust matrix, then the Eigentrust algorithm yields the right eigenvector associated to the Perron-Frobenius eigenvalue of $\mathbf{C}$.
\end{proposition}

\section{Trust Measuring in the Max-Plus Algebra}
\label{sec:MaxTrust}

Eigentrust does not distinguish the lack of interactions between two nodes, from the impossibility of such interactions due to the network's topology. We redefine $\mathbf{C}=(C_{ij})$ over the corresponding set of edges ($E$) to differentiate between these two cases, as follows:

\begin{equation}
\label{eq:epsilon_notation}
C_{ij} = 
\begin{cases}
c_{ij}       & \quad \text{if } (i,j)\in E\\
-\infty  & \quad \text{otherwise}
\end{cases}
\end{equation}

This differentiation gives rise to a distinct algebraic structure that facilitates the measuring of trust. The algebraic structure is an idempotent, commutative semiring (dioid) known as the \textit{Max-Plus Algebra} \cite{gaubert1997methods}, or \textit{max-plus} for short. Within max-plus we extend the Eigentrust algorithm, to cater for multi-agent systems with reducible trust matrices. To this end, we must consider how basic operations can be performed within max-plus, before examining trust within the algebra.

\subsection{Summing and Multiplying in Max-Plus}

\begin{definition}[Max-Plus \cite{gaubert1997methods}]
	Let $\mathbb{R}_{max}=\mathbb{R} \cup \{\varepsilon\}$ be the union of the set of real numbers $\mathbb{R}$ and $\varepsilon=-\infty$. Given $x,y \in  \mathbb{R}_{max}$, we define the following two operations.
	
	\begin{eqnarray}
	x\oplus y&=&\max(x,y) \\ \nonumber
	x\otimes y&=& x+y
	\end{eqnarray}
	
	\noindent  The set $\left(\mathbb{R}_{max}, \oplus, \otimes\right)$ constitutes a semiring commonly known as the \emph{Max-Plus Algebra}. 
\end{definition}

Since $x\oplus \varepsilon= x$ and $x\otimes 0=x$ for every $x\in \mathbb{R}_{max}$, $\varepsilon$ and $0$ are the neutral elements of the $\oplus$ and $\otimes$ operations  respectively. The term $e$ is preferred for referring to the latter, as to avoid confusion with $0\in \mathbb{R}$. Also note that the Max-Plus Algebra is an idempotent semiring in relation to $\oplus$, as $x\oplus x=x$ for any $x\in \mathbb{R}_{max}$.

Addition and multiplication in max-plus can be naturally extended to matrices by replacing the usual $``+"$ and $``\cdot"$ operators, with $\oplus$ and $\otimes$. The $m\times n$ \emph{zero matrix} is denoted by $\mathcal{E}$, such that $\mathcal{E}_{ij}=\varepsilon$ for all $i,j$. The $n\times n$ \emph{identity matrix}, $\mathbf{E}_{n}$, takes the form: 

\begin{equation*}
[E_{n}]_{ij}=\left\{ \begin{array}{c}
e \hspace{0.2cm} \text{ if } \hspace{0.2cm} i=j\\
\varepsilon \hspace{0.2cm} \text{ if } \hspace{0.2cm} i\neq j\\
\end{array}
\right.
\end{equation*}

The power of a matrix $\mathbf{A} \in \mathbb{R}_{max}^{n\times n}$ is inductively defined as $\mathbf{A}^{\otimes^{0}}\equiv\mathbf{E}_{n}$, and $\mathbf{A}^{\otimes^{k}}= \mathbf{A}\otimes \mathbf{A}^{\otimes^{k-1}}$ for $k>0$. Eigenvalues and eigenvectors can be described within the Max-Plus Algebra as follows.

\begin{definition}[Eigenvalues and eigenvectors]
	Let $\mathbf{A}\in\mathbb{R}_{max}^{n\times n}$, and consider the scalars $\lambda\in\mathbb{R}_{max}$, and vectors $\mathbf{v}\neq (\varepsilon,\varepsilon,\dots,\varepsilon) \in\mathbb{R}_{max}^n$ satisfying:
	\begin{equation}
	\label{eq:maxplus-eigenvalue}
	\mathbf{A}\otimes \mathbf{v}=\lambda \otimes \mathbf{v}
	\end{equation}
	$\lambda$ and $\mathbf{v}$ are referred to as the \textit{eigenvalues} and \textit{eigenvectors} of $\mathbf{A}$, respectively.
\end{definition}

Eigenvalues may be equal to $\varepsilon$. Both eigenvalues and eigenvectors may not be unique. The derivation of eigenvectors, i.e., the solutions to Equation \eqref{eq:maxplus-eigenvalue}, can be expressed more readily through  a linear optimization problem: $\max_j(a_{ij}+v_j)=\lambda+v_i$, where $\mathbf{A}=(a_{ij})$ and $\mathbf{v}=(v_{1},v_{2},\dots,v_{n})$.

Within the Max-Plus Algebra, the Perron-Frobenius Theorem takes on a more succinct form:

\begin{theorem}[Perron-Frobenius Theorem in Max-Plus \cite{akian1998asymptotics}]
\label{theo:perron_frobenius_max_plus}
	An \emph{irreducible} matrix $\mathbf{A}\in\mathbb{R}_{max}^{n\times n}$ has a unique dominant eigenvalue such that:
	\[\lambda_0=\bigoplus_{i=1}^{n}tr(A^{i})^{1/i}\]
\end{theorem}

Within our adapted version of Eigentrust, we will seek to find eigenvalues for reducible matrices. Such matrices can be rewritten in max-plus in \emph{normal form}.

\begin{definition}[Normal Form \cite{seneta2006non}]
	\label{eq:reducible_normal_form}
	Let $\mathbf{A}\in\mathbb{R}_{max}^{n\times n}$ be a reducible matrix, then its \textit{normal form} is the upper triangular matrix:
	
	\begin{equation}
	\label{eq:normal_form}
	\mathbf{A}=\left(
	\begin{array}{*5{c}}
	\mathbf{A}_{11}  &   \mathbf{A}_{12}  &   \hdots  &   \hdots  &   \mathbf{A}_{1n}\\
	\mathcal{E}  &   \mathbf{A}_{22}  &   \hdots  &   \hdots  &   \mathbf{A}_{2n}\\
	\mathcal{E} & \mathcal{E} &   \mathbf{A}_{33}  &   \hdots &   \vdots\\
	\vdots  &   \vdots  &   \vdots  &   \ddots  & \vdots\\
	\mathcal{E} & \mathcal{E}   &   \mathcal{E}  &   \hdots &  \mathbf{A}_{nn}\\
	\end{array}
	\right)
	\end{equation}
	
\noindent where $\mathbf{A}_{nn}$  is irreducible and the matrices $\mathbf{A}_{ii}$ are either irreducible or equal to $\varepsilon$, for all $1 \leq i \leq n$. The remaining block matrices that appear in Equation \eqref{eq:normal_form} are all different from $\mathcal{E}$. 
\end{definition}

\subsection{Measuring Trust in Max-Plus}

Consider line \ref{alg:eigentrust:op1} of Algorithm \ref{algo:eigentrust}. This operation updates the vector of trust values $\mathbf{t}$, effectively describing the evolution of a discrete system throughout $k$ iterations:
\begin{equation}
\label{eq:system}
\mathbf{t}(k + 1)=\mathbf{C}^{T}\mathbf{t}(k)
\end{equation}

If $\mathbf{C}$ is irreducible, the above equation cannot be further expanded, and Eigentrust would yield a satisfactory result on account of the Perron-Frobenius Theorem. The reducible case, on the other hand, leads to a more elaborate recurrence relation hindering Eigentrust's performance. Given $\mathbf{D}\equiv\mathbf{C}^{T}$ and considering its normal form we can rewrite Equation \eqref{eq:system} as \cite{konigsberg2009generalized}:

\begin{equation}
\label{eq:reducible_system}
\mathbf{t}(k + 1)=\mathbf{D}_{ii}\otimes \mathbf{t}_i(k) \oplus \bigoplus_{j=i+1}^{q}\mathbf{D}_{ij} \otimes \mathbf{t}_j(k), \forall k\leq 0
\end{equation}
where $\mathbf{D}_{ii}$ are irreducible or equal to $\epsilon$, for $i\leq n$; and $\mathbf{D}_{ij}\neq\mathcal{E}$, for $j=i+1$, $i\in\{0,1,\hdots, n-1\}$.

Provided $\mathbf{D}$ is reducible, there exist finite vectors $v_1, v_2, \hdots, v_n\in \mathbb{R}_{max}^{n\times 1}$ and scalars $\xi_1,\xi_2,\hdots, \xi_n\in\mathbb{R}$ producing a solution to Equation \eqref{eq:reducible_system}. More specifically, following \cite{konigsberg2009generalized}:

\begin{theorem}[Solution to a Reducible Dynamic System]
	\label{theo:reducible_system}
	The solution to the discrete dynamic system in equation \eqref{eq:reducible_system}, is given by:
	\begin{equation}
	\label{eq:max_trust_vector}
	t_i(k)=v_i\otimes\xi_i^{\otimes k}
	\end{equation}
	for all $k\geq0$ and $i\in\{1,2,\hdots,n\}$. The vectors $v_1, v_2, \hdots, v_n\in \mathbb{R}_{max}^{n\times 1}$ are finite, and the scalars $\xi_1,\xi_2,\hdots, \xi_n\in \mathbb{R}$ can be derived from the eigenvalues $\lambda_i$ of the irreducible block matrices $\mathbf{D}_{ij}$:
	\begin{equation}
	\label{eq:reducible_eigenvalues}
	\xi_i=\bigoplus_{j\in\mathcal{H}}\xi_j\oplus\lambda_j
	\end{equation}
	where $\mathcal{H}=\{j\in\{1,2,\hdots,n\}:j>i,\mathbf{D}_{ij}\neq\mathcal{E}\}.$
\end{theorem} 

\begin{proof}
Set $\mathbf{r}=\mathbf{w}\otimes\xi$ as our initial vector of trust values for $\xi\in\mathbb{R}^n$ and a random vector $\mathbf{w}\in\mathbb{R}_{max}^n$, and also let $\mathbf{D}_{ii}$ be a matrix. $\mathbf{D}_{ii}$ is irreducible, hence Theorem \ref{theo:perron_frobenius_max_plus} guarantees the existence of an eigenvalue $\lambda\in\mathbb{R}$ with eigenvector $\mathbf{v}\in\mathbb{R}^n$, chosen in accordance with the initial vector of trust scores: $\mathbf{v}\otimes\lambda\geq\bigoplus_{j=i+1}^{q}\mathbf{D}_{ij}\mathbf{r}_j$. Whenever $\ > \xi_j$ for $j\in\{i+1,\ldots,q\}$, $\mathbf{v}$ satisfies both
{\small
\[\mathbf{v}\otimes\lambda^{\otimes k+1}=\mathbf{D}_{ii}\otimes\mathbf{v}\otimes\lambda^{\otimes k}, \small\mathbf{v}\otimes\lambda^{\otimes k}\geq\bigoplus_{j=i+1}^{q}\mathbf{D}_{ij}\mathbf{w}_j\otimes\xi_j^{\otimes k}\]}
which, in turn, implies that
\[\mathbf{v}\otimes\lambda^{\otimes k+1}= max\{\mathbf{D}_{ii}\otimes\mathbf{v}\otimes\lambda^{\otimes k}, \bigoplus_{j=i+1}^{q}\mathbf{D}_{ij}\mathbf{w}_j\otimes\xi_j^{\otimes k}\}\]
or, equivalently, if we set $\mathbf{t}(k)\equiv\mathbf{v}\otimes\lambda^{\otimes k}$ equation \eqref{eq:reducible_system} is obtained. Note that this procedure is also applicable if the diagonal blocks are scalars by making $\lambda=\varepsilon$ and $\mathbf{v}=\bigoplus_{j=i+1}^{q}\mathbf{D}_{ij}\mathbf{w}_j$. 

When $\lambda \leq \xi_j$ for $j\in\{i+1,\ldots, q\}$ we could still obtain $\mathbf{v}$ as before given that $\bigoplus_{j=i+1}^{q}\mathbf{D}_{ij}\mathbf{w}_j$ has at least one finite element. This choice, however, would also involve 
\[\mathbf{v}\otimes\bigoplus_{j=i+1}^{q}\xi_j\geq \mathbf{D}_{ii}\otimes\mathbf{v}\otimes\lambda\oplus \bigoplus_{j=i+1}^{q}\mathbf{D}_{ij}\mathbf{w}_j\xi_j\]
which again leads to equation \eqref{eq:reducible_system}, if we let $\mathbf{t}(k)\equiv\mathbf{v}\otimes\mathbf{\mu}^{\otimes k}$ and 	$\mu_i=\bigoplus_{j\in\mathcal{H}}\xi_j\oplus\lambda_j$ for all $k\geq0$, $i\in\{1,2,\hdots,n\}$, and $\mathcal{H}$ defined as in the statement of the theorem.
\end{proof} 

\begin{algorithm}[H]
\caption{Power Method for regular irreducible matrices in Max-Plus}
\label{algo:maxpower}
\begin{algorithmic}[1]
\Require $\mathbf{r}$: Arbitrary vector of trust values, $\mathbf{C}$: Trust Matrix.
\Ensure $\lambda$: Eigenvalue of $\mathbf{C}$, $\mathbf{v}$: Eigenvector of $\mathbf{C}$.
\Procedure{max\_power}{}
\State $p\gets 0$
\State $\mathbf{v}_p\gets \mathbf{r}$
\Repeat
\State $\mathbf{v}_{p+1}\gets \mathbf{C}'\mathbf{v}_p$
\State $p\gets p+1$
\Until{$\mathbf{v}_q\equiv c\otimes\mathbf{v}_p \land c\geq 0$}
\State $\lambda\gets \frac{c}{p-q}$
\State $\mathbf{v}\gets 
\bigoplus_{i=1}^{p-q}\left(\lambda^{\otimes(p-q-i)}\otimes\mathbf{v}_{q+i-1}\right)$\\
\Return $\lambda, \mathbf{v}$
\EndProcedure
\end{algorithmic}
\end{algorithm}

Theorem \ref{theo:reducible_system} indicates that trust can be measured over reducible matrices as prescribed by the Eigentrust algorithm, invoking the spectral properties of the irreducible blocks of its normal form. In graph theoretical terms, this means that we can recover the main asymptotic traits of the system by looking into the connected components of the underlying network. Based on this result we introduce Algorithm \ref{algo:MaxTrust}, referred to as \emph{MaxTrust}.

\begin{center}
\centering
\begin{algorithm}
\caption{Trust-Measuring Algorithm in Max-Plus}
\label{algo:MaxTrust}
\begin{algorithmic}[1]
\Require $\mathbf{C}$: Regular Trust Matrix, $\mathbf{w}$: Vector of initial trust values, $T$: Terminal time.
\Ensure $\mathbf{t}$: Rank of agents at terminal time.
\Procedure{$MaxTrust$}{}
\State $\mathbf{D}\gets$ \Call{get\_normal\_form}{$\mathbf{C}$}
\State $\lambda_n, \mathbf{v}_n\gets $ \Call{max\_power}{$\mathbf{D}_{nn}$} \label{alg:MaxTrust:op0}
\State $\xi_n\gets \lambda_n$
\State $j\gets n-1$;
\While{$j>1$}
\State $\lambda_j\gets$ \Call{max\_power}{$\mathbf{D}_{jj}$}.
\If{$\lambda_{j}>\xi_{j+1}$} \label{alg:MaxTrust:op1}
\State $\xi_j\gets \lambda_{j}$ 
\State $\mathbf{v}_j\gets \bigoplus_{k=1}^{n}\mathbf{D}_{jk}\otimes \mathbf{w}_{k} \otimes \lambda_j^{\otimes j-1}$ \label{alg:MaxTrust:op2}
\Else
\State $\xi_j\gets \lambda_{j+1}$
\State $\mathbf{v}_j\gets (\xi_j)^{-1}\otimes\bigoplus_{k=1}^{n}\mathbf{D}_{jk}\otimes \mathbf{w}_{k} \otimes \lambda_j^{\otimes j-1}$ \label{alg:MaxTrust:op3}
\EndIf
\State $j\gets j-1$
\EndWhile
\Return $\mathbf{t}\gets \mathbf{v}\otimes\mathbf{\xi}^{\otimes T}$ \label{alg:MaxTrust:op4}
\EndProcedure
\end{algorithmic}
\end{algorithm}
\end{center}

As with Eigentrust, the vector of initial trust values $\mathbf{w}$ can be selected randomly. After converting the trust matrix to normal form, a max-plus adaptation of the Power Method (Algorithm \ref{algo:maxpower}) is applied to the last irreducible block of $\mathbf{D}$, so as to obtain its corresponding eigenvalues and eigenvectors. A similar operation is carried out for the rest of the diagonal blocks in $\mathbf{D}$, serving as the basis for the eigenvectors of the supradiagonal blocks (lines \ref{alg:MaxTrust:op2} and \ref{alg:MaxTrust:op3}). Finally, note that the expressions in lines \ref{alg:MaxTrust:op1} and \ref{alg:MaxTrust:op4} mirror Equation \eqref{eq:reducible_eigenvalues} in the conventional algebra and Equation \eqref{eq:max_trust_vector} in max-plus, respectively. 

\section{Evaluation}
\label{sec:eval}

We performed an empirical evaluation of our algorithm, which we refer to as MaxTrust, comparing it against Eigentrust in a simple simulated computer peer-to-peer network scenario. We begin by describing our experimental setup before detailing our results.

\subsection{Experimental Setup}
Our experiments evaluated how trust propagates across a peer-to-peer network of routers, whose goal is to make routing decisions for data by deciding which of their neighbours such data should be transmitted to. Routers in the network interact with each other by exchanging connectivity information, consisting of trust measures regarding the network. This trust measure mirrors $t_{ij}$ --- if router $i$ is broadcasting such a trust measure to its neighbours, index $j$ of the vector will contain either the level of trust $i$ ascribes to $j$; or  $\varepsilon$ if $i$ has no knowledge of $j$ through direct or indirect experiences (in the case of MaxTrust), or 0 (in the case of Eigentrust).

Within the network, routers could either be trustworthy or malicious. The former broadcast trust measures correctly, while the latter transmit either a 0 or a value which begins at 0.5 and decays towards 0 as the router repeatedly interacts with others (to simulate the malicious router trying to undermine the network more actively). During experiment initalisation, each trustworthy router began by imputing a random level of trust to all of its neighbours (with $\varepsilon$ in the remaining indices of its trust vector). 

Each experiment was run over 100 time steps. Within each time step, all routers were given 10 opportunities to exchange information with their most trusted neighbour. After each such interaction, the trust they ascribed to their neighbour was either increased  (if the neighbour was a trustworthy router) or decreased (if the neighbour was a malicious router) by 0.0001. There was also a 0.0025 chance of trust decreasing (effectively due to a mis-catagorisation of the neighbour).

After each time step, all routers computed new trust values for the system using  Eigentrust or MaxTrust, and the process repeated. We considered three different router topologies, namely free trees (with a branching factor of 2); a toroidal network; and a random network of connections. Each network began by containing 4 routers with 8 links between them (generated according to the topology).

Each topology was then evaluated under 3 different scenarios.
\begin{itemize}
\item \textbf{Scenario 1} The network remained unchanged over all 100 trials.
\item \textbf{Scenario 2} Every 5 interactions, between 2 and 6 new routers were added to the system. Half of the new routers in the system were set to be malicious. 
\item \textbf{Scenario 3} As in Scenario 2, between 2 and 6 new routers were added to the system every 5 interactions. None of the routers in the initial system were malicious, each new router had a 1/3 likelihood of being malicious.
\end{itemize}

\noindent In all cases, routers were added in such a way so as to preserve the network's topology.

\subsection{Results}
We ran a total of 18 experiments (for each topology, scenario and trust algorithm combination), averaging 100 runs of each experimental condition (over 100 time steps) to obtain the results shown in Figure \ref{fig:convergence}. The vertical axis in each plot compares averaged distance between the dominant eigenvector obtained by Eigentrust or MaxTrust ($v$) with the dominant eigenvector computed for the trust matrix obtained at the end of each experimental run ($v_{\lambda_0}$). This latter dominant eigenvector was obtained from the corresponding eigenvalue computed with
the Newton method \cite{wilkinson1965algebraic}.

\begin{figure}[t]
	\centering
	\includegraphics[trim=0 90 0 105,clip,width=\columnwidth]{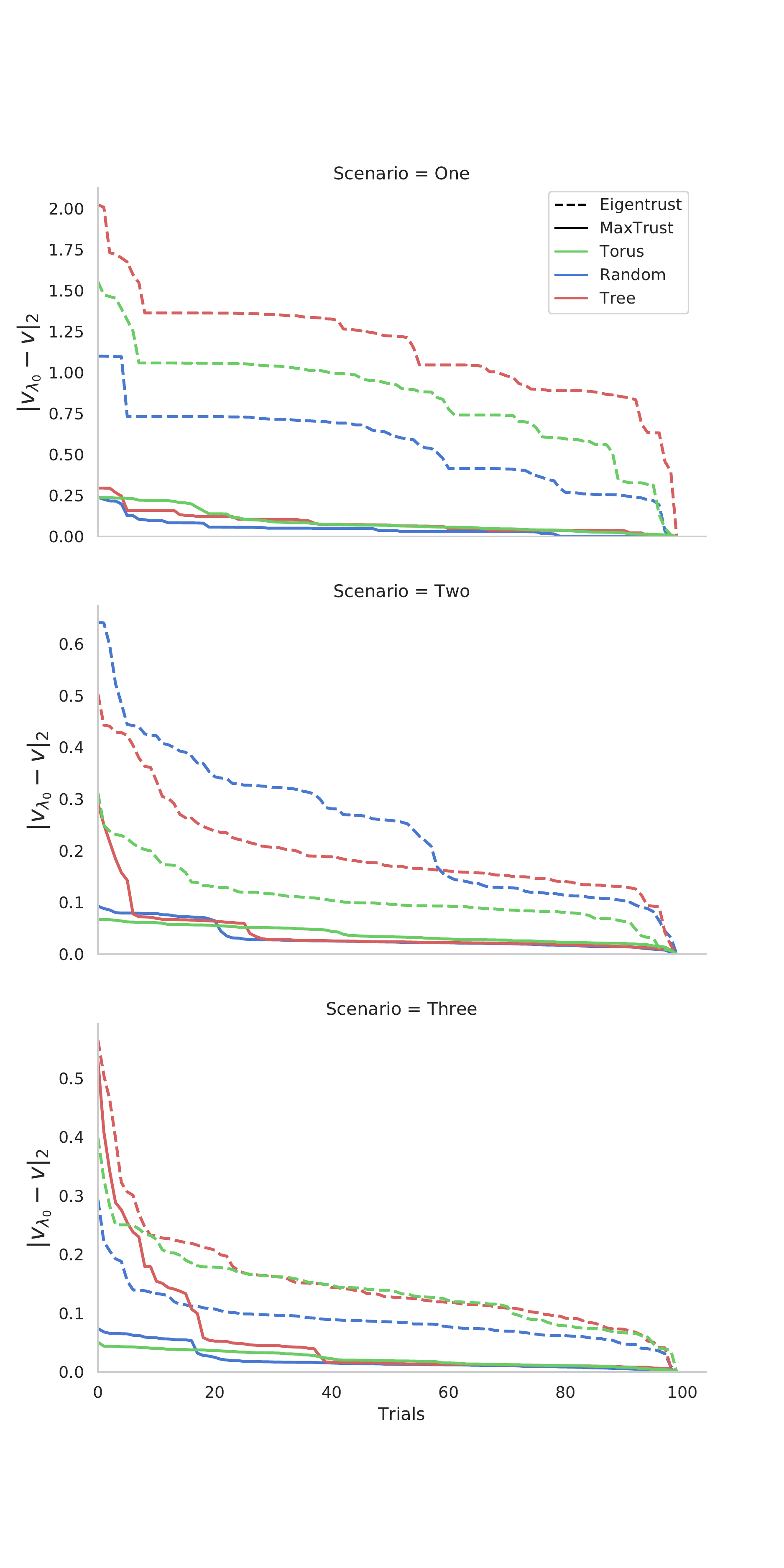}\\ 
	\caption{Relative Convergence to the Dominant Eigenvalue.}
	\label{fig:convergence}
\end{figure}

\begin{table}[t]
	\centering
	\resizebox{8cm}{2.5cm}{
			\begin{tabular}{|P{2cm}|M{1.8cm}|M{2cm}|M{2cm}|M{2.5cm}|}
				\hline
				\textbf{Scenario} & \textbf{Network Structure} & \textbf{Mean} & \textbf{Standard Deviation} & \textbf{95\% HDP} \\ 
				\hline
				\multirow{2}{*}{\parbox{.8\linewidth}{\centering\vspace{1cm}\textbf{One}}}
				& Random MT & 0.035 &  0.001 &  [0.034, 0.037]\\
				& Random ET & 0.516 &  0.039 &  [0.442, 0.583]\\ \cline{2-5} 
				& Torus MT & 0.074 &  0.003 &  [0.067, 0.079]\\  
				& Torus ET & 0.802 &  0.066 &  [0.662, 0.925]\\ \cline{2-5} 
				& Tree MT & 0.068 &  0.0015 &  [0.065, 0.070 ]\\ 
				& Tree ET &  1.106 &  0.059 & [1.012, 1.225]\\ \hline
				\multirow{2}{*}{\parbox{.8\linewidth}{\centering\vspace{1cm}\textbf{Two}}}
				& Random MT & 0.028 &  0.0003 &  [0.027, 0.028]\\  
				& Random ET & 0.221 &  0.012 & [0.203, 0.246]\\ \cline{2-5} 
				& Torus MT & 0.036 &  0.0001 &  [0.035, 0.036]\\ 
				& Torus ET & 0.097 &  0.0012 & [0.095, 0.099]\\ \cline{2-5} 
				& Tree MT & 0.029 &  0.0002 &  [0.029, 0.030]\\ 
				& Tree ET & 0.176 &  0.0029 &  [0.171, 0.180]\\ \hline
				\multirow{2}{*}{\parbox{.8\linewidth}{\centering\vspace{1cm}\textbf{Three}}}
				& Random MT & 0.016 &  0.0002 & [0.016, 0.017]\\ 
				& Random ET & 0.080 &  0.0006 &  [0.078, 0.080]\\ \cline{2-5} 
				& Torus MT & 0.020 &  0.0001 & [0.020, 0.020]\\
				& Torus ET & 0.125 &  0.0019 &  [0.121,  0.129]\\ \cline{2-5} 
				& Tree MT & 0.031 &  0.001 & [0.028, 0.033]\\ 
				& Tree ET & 0.129 &  0.002 &  [0.125, 0.135]\\ \hline
			\end{tabular}
	}
			\caption{Summary of results for MaxTrust (MT) and Eigentrust (ET) over all scenarios and topologies.}
	\label{tab:test}
\end{table}

While all methods converge to the final dominant eigenvector as more information is exchanged between routers, it is clear that by differentiating between distrust (i.e., a trust value of 0) and no trust information (i.e., a trust value of $\varepsilon$), MaxTrust significantly outperforms Eigentrust across all topologies and scenarios. Within Scenario 1, trees gave the sparsest connectivity structure, meaning that Eigentrust struggled most in this case.  As connectivity increased in Scenario 1, Eigentrust's performance improved, but still converged much more slowly than MaxTrust. 

Within Scenario 2, Tori provided more connectivity than trees, again leading to better performance for Eigentrust in the former case, and both outperforming random networks. Given the results of the third scenario, we believe that this behaviour is provoked by the disruption in the transmission of information caused by the introduction of malicious routers. Indeed, when the majority of routers are not malicious, as in Scenario 3, the extra connectivity provided by random networks enhance their performance, while trees and tori induce no considerable changes. 

Table \ref{tab:test} provides further evidence of these observations. It summarises the essential statistics detailing the posterior distributions of the distance $|v_{\lambda_0}-v|$ for MaxTrust (MT) and Eigentrust (ET). When treated as random variables the distributions of such deviations indicate how different the two results may be. Here, the mean and standard deviation are calculated over all 100 time steps (again over all 100 experiment runs). The  low means and standard deviations of MaxTrust  across all scenarios and topologies demonstrate its faster rate of convergence when compared to Eigentrust.

Ultimately, our results demonstrate the benefit of a trust and reputation system being able to differentiate between the lack of trust in an agent (i.e., a 0 trust value), and the lack of information about trust in an agent (captured via $\varepsilon$). While Eigentrust conflates these two concepts, MaxTrust deals with them separately. The differentiation between the two concepts is particularly important in open dynamic multi-agent systems.

\section{Conclusions}
Algorithms for trust measurement and computation are critical for the effective operation of open multi-agent systems. Due to its simplicity and effectiveness, Eigentrust is perhaps the most widely used trust-measuring algorithm. Building on an analysis of the situations where Eigentrust performs poorly (notably in cases where no connected components exist), we introduced the MaxTrust algorithm. This algorithm shares the same basic intuitions used to create Eigentrust, but builds on the Max-Plus Algebra, and in doing so, provides improved convergence to the \textit{ex-post} or actual trust values when compared to Eigentrust.

In this work we considered only the simplest (non-distributed) form of Eigentrust. As future work, we intend to extend MaxTrust to deal with distributed trust update; we believe that this extension can be performed in a manner similar to that in which Eigentrust's counterpart was obtained, but not without first evaluating MaxTrust for such a case. 

We also intend to investigate the theoretical properties of MaxTrust in future work. Such properties include identifying guarantees on convergence rates, and the effects of different attacks against the algorithm. Given the shared intuitions between MaxTrust and Eigentrust, we believe that many results will carry through, but stress that a theoretical and empirical evaluation of MaxTrust under different scenarios is critical if the performance improvements, it seems to hold, are to be realised in practical applications.
 

\bibliographystyle{named}
\bibliography{bibliography}

\end{document}